\documentclass[11pt, onecolumn]{IEEEtran}
\usepackage[usenames]{color}
\usepackage{times}
\usepackage{algorithmic}
\usepackage{amsmath}
\usepackage{amssymb}
\usepackage{amsthm}
\usepackage{dsfont}
\usepackage{color}
\usepackage{colortbl}
\usepackage{float}
\usepackage{cite}
\usepackage{verbatim}
\usepackage[pdftex]{graphicx}
\usepackage{caption}
\usepackage{subcaption}

\newcommand{\field}[1]{\mathbb{#1}}

\newcommand{\F}{\field{F}}

\newcommand{\cC}{{\cal C}}
\newcommand{\cD}{{\cal D}}
\newcommand{\cG}{{\cal G}}

\newcommand{\cM}{{\cal M}}

\newcommand{\cP}{{\cal P}}

\newcommand{\sP}{\cP}
\newcommand{\sG}{\cG}

\newcommand{\Gr}{\smash{{\sG\kern-1.5pt}_q\kern-0.5pt(n,k)}}
\newcommand{\Gfourk}{\smash{{\sG\kern-1.5pt}_q\kern-0.5pt(4k,2k)}}
\newcommand{\Gk}{\smash{{\sG\kern-1.5pt}_q\kern-0.5pt(n,k_1)}}
\newcommand{\Gkk}{\smash{{\sG\kern-1.5pt}_q\kern-0.5pt(n,k_2)}}
\newcommand{\Grtwo}{\smash{{\sG\kern-1.5pt}_2\kern-0.5pt(n,k)}}
\newcommand{\Gkone}{\smash{{\sG\kern-1.5pt}_q\kern-0.5pt(n,k_1)}}
\newcommand{\Gktwo}{\smash{{\sG\kern-1.5pt}_q\kern-0.5pt(n,k_2)}}
\newcommand{\Ps}{\smash{{\sP\kern-2.0pt}_q\kern-0.5pt(n)}}

\newcommand{\gammav}{\hbox{\boldmath$\gamma$}}

\newcommand{\cv}{{\bf c}}

\newcommand{\fv}{{\bf f}}

\newcommand{\uv}{{\bf u}}

\newcommand{\vv}{{\bf v}}

\newcommand{\Mm}{{\bf M}}

\newcommand{\Um}{{\bf U}}

\newcommand{\Vm}{{\bf V}}

\newcommand{\tauv}{\hbox{\boldmath$\tau$}}

\newcommand\floorb[1]{\left\lfloor #1 \right\rfloor}

\newtheorem{theorem}{Theorem}
\newtheorem{corollary}{Corollary}
\newtheorem{lemma}[theorem]{Lemma}

\newtheorem{remark}{Remark}
\newtheorem{example}{Example}

\begin{document}

\bibliographystyle{IEEEtran}
\title{Error-Correcting Regenerating \\ and Locally Repairable Codes\\via Rank-Metric Codes}
\author{
\IEEEauthorblockN{Natalia Silberstein, Ankit Singh Rawat and Sriram Vishwanath}
\thanks{N.~Silberstein is with the Department of Computer Science, Technion --- Israel Institute of Technology, Haifa 32000, Israel (email: natalys@cs.technion.ac.il).}
\thanks{A.~S.~Rawat and S.~Vishwanath are with the Department of Electrical and Computer Engineering, The University of Texas at Austin, Austin, TX 78712, USA (e-mail: ankitsr@utexas.edu, sriram@austin.utexas.edu).}
\thanks{This paper was presented in part at the 2012 50th Annual
Allerton Conference on Communication, Control, and Computing.}
}

\maketitle
\begin{abstract}
This paper presents and analyzes a novel concatenated coding scheme for enabling error
resilience in two distributed storage settings: one being storage using
existing regenerating codes and the second being storage using locally
repairable codes. The concatenated coding scheme brings together a maximum
rank distance (MRD) code as an outer code and either a globally
regenerating or a locally repairable code as an inner code.
Also, error resilience for combination of locally repairable codes with
regenerating codes is considered.
This concatenated coding system is designed to handle two
different types of adversarial errors: the first type includes
an adversary that can replace the content of an affected node
only once; while the second type studies an adversary that
is capable of polluting data an unbounded number of times.
The paper establishes an upper bound on the resilience capacity for
a locally repairable code and proves that this concatenated coding
coding scheme attains the upper bound on resilience capacity for
the first type of adversary.
Further, the paper presents mechanisms
that combine the presented concatenated coding scheme with subspace signatures
 to achieve error resilience for the second type of~errors.


\end{abstract}
\section{Introduction}
\label{sec:introduction}

Distrubuted storage systems (DSS) have gained importance over recent years as  dependable, easily accessible and well administrated cloud resource for both individuals and businesses. There are multiple research issues that are unique to DSS; some of which are logistical and market driven, while others relate with their underlying design. A primary concern in designing DSS is to ensure resilience to failures, as it is desirable that a user (\emph{data collector}) can retrieve the stored data even in the presence of node failures. As studied in the pioneering work by Dimakis et al.~\cite{dimakis}, coding introduces redundancy to a storage system in the most efficient manner to enable resilience to failures.
In \cite{dimakis}, Dimakis et al. go one step further: when a single node fails, they propose reconstructing the data stored on the failed node in order to maintain the required level of redundancy in the system.  This process of data reconstruction for a failed node is called the \emph{node repair process}~\cite{dimakis}. During a node repair process,  the node which is added to the system to replace the failed node, downloads  data from a set of surviving nodes to reconstruct the lost (or its equivalent) data.

{\em Regenerating codes} and {\em locally repairable codes (LRCs)} are two families of codes that are especially designed to allow for efficient node repairs in DSS. In particular, regenerating codes are designed to reduce \emph{repair bandwidth}, i.e., the amount of data downloaded from surviving nodes during the node repair process. On the other hand, LRCs are designed to have a small number of nodes participating in the node repair process. The constructions for these two families of codes can be found in~\cite{dimakis, WuDim09, SRKR_itw10, SuhRam_isit10, RSK11, zigzag13, DRWS11, DaOg11, pyramid, HL-M2007, oggier_proj, oggier_hom, HSXOCGLY12, DimDim12, PKLK12, RaVi12, SRV12, RKSV12,  Gopalan12, KPLV12, SAPDVCB13, Hollman13} and references therein.

Although failure resilience is of critical importance as failure of storage nodes is commonplace in storage systems, there are multiple other design considerations that merit study in conjunction with failure resilience. These include security, error resilience, update efficiency and load balancing. In this paper, we address the issue of instilling error resilience in DSS, particularly against adversarial errors. In particular, we model and present coding methodologies that allow for a data collector to correctly decode data even in presence of adversarial errors.

In~\cite{dimakis}, Dimakis et al. establish close connections between DSS with node repairs and the network coding problem. Thus, it is natural to apply the  techniques used for error correction in network coding for DSS with node repairs. Rank-metric codes are known to be a powerful solution to error correction problem in network coding~\cite{SKK08,EtSi09,NoUF10,SK11,NZGYS12}. The primary idea of this paper is to apply a similar technique, based on rank-metric codes, for error resilience in DSS.

In this paper, we propose a novel concatenated coding scheme, based on rank-metric codes, which provides resilience against  adversarial errors in both regenerating codes and LRCs  based DSS. Moreover, we also consider error resilience in DSS employing combination of LRCs and regenerating codes~\cite{RKSV12, KPLV12}.

The problem of reliability of DSS against adversarial errors was considered in~\cite{DiDiHo10,oggier_byzantine,pawar11, RSRK_isit12}. In particular, in \cite{pawar11}, Pawar et al.  derive upper bounds on the amount of data that can be stored on the system and reliably made available to a data collector when bandwidth optimal node repair is performed, and present coding strategies that achieve the upper bound for a particular range of system parameters, namely in the bandwidth-limited regime. A related but different problem of securing stored data against passive eavesdroppers is addressed in \cite{pawar11, SRK_globecom11, RKSV12, GRCP13}.

In this paper, we study the notion of an {\em omniscient adversary} who can observe all nodes and has full knowledge of the coding scheme employed by the system~\cite{pawar11}. As in~\cite{pawar11}, we assume an upper bound on the number of nodes that can be controlled by such an omniscient adversary. We classify adversarial attacks by such an adversary into two classes:
\begin{enumerate}
\item {\em Static errors:}  an omniscient adversary replaces the content of an affected node with nonsensical (unrelated) information \emph{only once}. The affected node 
     uses this \emph{same} polluted information during all subsequent repair and data collection processes. Static errors represent a common type of data corruption due to wear out of storage devices, such as latent disk errors or other physical defects of the storage media, where the data stored on a node is permanently distorted. 
\item {\em Dynamic errors:} an omniscient adversary may replace the content of an affected node \emph{each time} the node is asked for its data during data collection or repair process. This kind of errors captures any malicious behaviour, hence is more difficult to manage in comparison with static errors.
\end{enumerate}

 We present a novel concatenated coding scheme for DSS which provides resilience against these two classes of attacks and allow for either optimal repair bandwidth or optimal local node repairs. In our scheme for an optimal repair bandwidth DSS, the content to be stored is first encoded using a maximum rank distance (MRD) code. The output of this outer code is further encoded using a regenerating code, which allows for  bandwidth efficient node repairs. Using an MRD code, which is an optimal rank-metric code, allows us to quantify the errors introduced in the system using their rank as opposed to their Hamming weights. The dynamic nature of the DSS causes a large number of nodes to get polluted even by a single erroneous node, as  false information spreads from node repairs. Thus, a single polluted node infects many others, resulting in an error vector with a large Hamming weight. Using rank-metric codes can  help alleviate this problem as the error that a data collector has to handle has a known rank, and can therefore be corrected by an MRD code with a sufficient rank distance.
 Using an $(n,k)$ bandwidth efficient MDS array code as inner code facilitates bandwidth efficient  node repair in the event of a single node failure and allows the data collector to recover the original data from any subset of $k$ storage nodes. In this paper, we use exact-regenerating linear bandwidth efficient codes operating at the minimum-storage regenerating (MSR) point~\cite{dimakis}.  However, our construction can be utilized for regenerating codes operating at any other point of repair bandwidth vs. per node storage trade-off.

 Further, we focus on error resilience for  locally repairable DSS. We present a bound on the resilience capacity for such DSS. To the best of our knowledge,  this is the first work which considers the issue of error resilience for locally repairable codes. Then we show that our scheme based on MRD codes can be applied to LRC codes as well. In addition, we consider error resilience in codes with optimal bandwidth \emph{local} regeneration, where MSR or MBR codes are utilized as local codes~\cite{RKSV12, KPLV12, MBR-LRC}.

Our coding schemes are directly applicable to the \textit{static error} model and are optimal in terms of amount of data that can be reliably stored on a DSS under static error model. The model with \textit{dynamic errors} is more complicated, as it permits a single malicious node to change its pollution pattern, and introduce an arbitrarily large error both in Hamming weight and in rank. For dynamic error model, we present two solutions based on the concatenated coding scheme employed for static error model. (We consider only the solutions for a secure regenerating code, as for a secure locally repairable code the same ideas can be applied). One solution is to exploit the inherent redundancy in the encoded data due to outer code, i.e., an MRD code, and perform error free node repair even in the presence of adversarial nodes. This solution, a na\"{\i}ve method, is optimal for a specific choice of parameters. An alternative solution combines our concatenated coding scheme with subspace signature based cryptographic schemes to control the amount (rank) of pollution (error) introduced by an adversarial node. We employ the signature scheme by Zhao et al. \cite{zhao}, which essentially reduces the dynamic error model to the one similar to a static error model, and helps us bound the rank of error introduced by an adversarial node. Note that hash function based solutions have previously been presented in the context of DSS to deal with errors \cite{DiDiHo10,pawar11}. While promising, these hash functions based approaches provide only probabilistic guarantees for pollution containment.

The rest of the paper is organized as follows: In Section~\ref{sec:preliminaries}, we start with a description of our system model and
provide a brief overview of rank-metric, regenerating and locally repairable codes.
In Section~III, we consider the static error model. First, we describe the construction of our storage scheme which allows for bandwidth efficient repairs and prove its error resilience. Second, we consider error resilience in LRCs. We provide a upper bound on the resiliency capacity and present modifications of our scheme which enable optimal local repairs.
In Section IV, we address the dynamic error model. We conclude with Section V, where we list a set of open problems.



\section{Preliminaries}
\label{sec:preliminaries}

\subsection{System model}
\label{subsec:model}
Let $\cM$ be the size of a file $\fv$ over a finite field $\F$ that needs to be stored on a DSS with $n$ nodes. Each node contains
$\alpha$ symbols over $\F$. A data collector reconstructs the original file $\fv$ by downloading the
data stored on any set of $k$ out of $n$ nodes. This property of a DSS is called \emph{`any $k$ out of $n$' property}, and we use $(n,k)$-DSS to represent a storage system which has `any $k$ out of $n$' property.

\subsection{Rank-metric codes}
\label{subsec:rank-metric codes}
All the constructions of coding schemes for distributed storage systems provided in this paper are based on a family of error correcting codes in rank metric, called Gabidulin codes.

Let $\F_{q^m}$ be an extension field of $\F_q$. Since $\F_{q^m}$
can be also considered as an $m$-dimensional vector space over $\F_q$, any element  $\gamma\in\F_{q^m}$ can be represented  as the vector ${\gammav=(\gamma_1,\ldots, \gamma_m)\in \F_q^m}$, such that $\gamma=\sum_{i=1}^{m}b_i\gamma_i$, for a fixed basis $\{b_1,\ldots, b_m\}$ of the field extension. Similarly, any vector $\bf{v}$ $=(v_1, \ldots, v_N)\in\F_{q^m}^N$ can be represented by an $m\times N$ matrix $\Vm=[v_{i,j}]$ over $\F_q$, where each entry $v_i$  of $\vv$ is expanded as a column vector $(v_{i,1},\ldots,v_{i,m})^T$.
The \emph{rank} of a vector $\vv\in\F_{q^m}^N$, denoted by $\rm{rank}(\vv)$, is defined as the rank of the $m\times N$ matrix $\Vm$ over~$\F_q$. Similarly, for two vectors $\vv,\uv \in \F_{q^m}^N$, the \emph{rank distance} is defined by $$d_R(\vv,\uv)=\rm{rank}(\Vm-\Um).$$
An $[N,K,D]_{q^m}$ \textmd{rank-metric code} $\cC\subseteq\F_{q^m}^N$ is a linear block code over $\F_{q^m}$ of length $N$, dimension $K$  and minimum rank distance $D$. A rank-metric code that attains the Singleton bound $D\leq N-K+1$ for the rank metric is called a \emph{maximum rank distance} (MRD) code~\cite{Del78,Gab85,Rot91}.
For $m\geq N$, a construction of MRD codes was presented by Gabidulin~\cite{Gab85}.
Similar to Reed-Solomon codes, Gabidulin codes can be obtained by
evaluation of polynomials; however, for Gabidulin codes, a special family of polynomials \emph{linearized polynomials} is used.
A linearized polynomial $f(x)$ over $\F_{q^m}$ of $q$-degree $t$ has the form $f(x)=\sum_{i=0}^{t}a_ix^{q^i}$,
where $a_i\in \F_{q^m}$, and $a_{t}\neq 0$.

\begin{remark} Note, that evaluation of a linearized polynomial is an $\F_{q}$-linear
transformation from $\F_{q^m}$ to itself, i.e., for any ${a, b \in \F_q}$ and
$\gamma_1, \gamma_2\in\F_{q^m}$, we have $f(a\gamma_1 + b\gamma_2)= af(\gamma_1)+ bf(\gamma_2)$~\cite{MWSl78}.
\end{remark}

A codeword of an $[N,K,D=N-K+1]_{q^m}$ Gabidulin code $\cC^{\rm{Gab}}$ is defined as $\cv = (f(g_1),f(g_2), \ldots, f(g_N))\in\F_{q^m}^N$, for $m\geq N$, where
$f(x)$ is a linearized polynomial over $\F_{q^m}$ of $q$-degree at most $K-1$ with $K$ message symbols as its coefficients, and the evaluation points $g_1,\ldots, g_N \in \F_{q^m}$ are linearly independent over $\F_q$~\cite{Gab85}.

\subsubsection{Rank Errors and Rank Erasures Correction}
\label{subsubsec:rank error}

Let $\cC^{\rm{Gab}}\subseteq \F_{q^m}^N$ be a Gabidulin code with the minimum distance~$D$.
Let $\mathbf{c}\in \cC^{\rm{Gab}}$ be the transmitted
codeword and let $\mathbf{r}=\mathbf{c}+\mathbf{e_{total}}$ be the received word.
The code $\cC^{\rm{Gab}}$ can correct any vector error of the form
\begin{align}
\label{eq:rank error}
\mathbf{e_{total}}=\mathbf{e_{error}}+\mathbf{e_{erasure}} = (e_1\mathbf{u_{1}}+ \ldots +e_t\mathbf{u_{t}})+ (r_1\mathbf{v_{1}}+ \ldots +r_s\mathbf{v_{s}}),
\end{align}
 as long as $2t+s\leq D-1$.
The  first part $\mathbf{e_{error}}$ is called  a \emph{rank error}
of rank $t$, where $e_i\in\F_{q^m}$ are linearly independent over the base field $\F_q$,
unknown to the decoder, and $\mathbf{u}_i\in\F_q^N$
are linearly independent vectors of length~$N$, unknown to the decoder.
The second part $\mathbf{e_{erasure}}$ is called a\emph{ rank erasure}, where
$r_i\in\F_{q^m}$ are linearly independent over the base field $\F_q$,
unknown to the decoder, and $\mathbf{v}_i\in\F_q^N$
are linearly independent vectors of length~$N$,  and known to the decoder. (In this paper we are interested only in rank \emph{column} erasures, so the rank \emph{row} erasures are not considered). \cite{Gab85,SiKs09} present decoding algorithms for rank-metric codes.

Note that an $[N, K, D]_{q^m}$ MRD code is in particular an MDS code over $\F_{q^m}$, and then can correct erasures of any $D-1$ symbols in a codeword. Such symbols erasures are in particular column erasures, when the codeword is considered in the matrix representation.

\subsection{Vector codes}
A linear $[n, \cM, d_{\min}, \alpha]_{\F}$ vector code $C$ of length $n$ over a finite field $\F$  is defined as a linear subspace of $\F^{\alpha n}$ of dimension $\cM$.
The symbols $\cv_i$, $1\leq i\leq n$, of a codeword $\cv\in C$ belong to $\F^{\alpha}$, i.e., are vectors or \emph{blocks} of size $\alpha$.
The minimum distance $d_{\min}$ of $C$ is defined as the minimum Hamming distance over $\F^{\alpha}$.

Vector codes are  also known as \emph{array }codes.
An $[n, \cM, d_{\min}, \alpha]_{\F}$ array code is called {\em MDS array code} if $\alpha|\cM$ and $d_{\min} = n - \frac{\cM}{\alpha} + 1$.
 Constructions for MDS array codes can be  found e.g. in~\cite{BlRo99, CaBr06}.

 In order to store a file $\fv\in \F^{\cM}$  on a DSS using a vector code $C$, $\fv$ is first encoded to a codeword $\cv = (\cv_1, \cv_2,\ldots, \cv_n) \in C$.
Each symbol (block) $\cv_i\in \F^{\alpha}$ of the codeword is then stored on a distinct node.
 \subsection{Regenerating codes}
 \label{subsec: regenerating codes}
 Regenerating codes are a family  of vectors codes for an $(n,k)$-DSS  that allow for efficient repair of failed nodes.
When a node fails, its content can be reconstructed by downloading $\beta\leq \alpha$ symbols from each node in a set of
$d$, $k\leq d\leq n-1$, surviving nodes.
We denote by $[n,\cM,d_{\min}=n-k+1,\alpha,\beta,d]_{\F}$ a linear regenerating code. Note that data can be reconstructed from any $k$ symbols of a codeword and therefore the minimum distance is $d_{\min}=n-k+1$.
A trade-off between storage per node $\alpha$ and {\em repair bandwidth} $\gamma \triangleq d\beta$ was established in~\cite{dimakis}.
 Two classes of codes that achieve two extreme points of this trade-off are known as {\em minimum storage regenerating (MSR)} codes and {\em minimum bandwidth regenerating (MBR)} codes. The parameters $(\alpha,\gamma)$ for MSR and MBR codes are given by $\left(\frac{\cM}{k},\frac{\cM d}{k(d-k+1)}\right)$ and $\left(\frac{2\cM d}{2kd-k^2+k},\frac{2\cM d}{2kd-k^2+k}\right)$, respectively~\cite{dimakis}.

 In this paper we focus on the family of linear MSR codes. Note that these codes are also MDS array codes. We refer to these codes as $[n,\cM=\alpha k,d_{\min}=n-k+1,\alpha,\beta=\frac{\alpha}{d-k+1},d]_{\F}$ optimal repair MDS array codes.

\begin{figure*}[!t]
        \centering
        \begin{subfigure}[b]{0.47\textwidth}
                \centering
                \includegraphics[width=\columnwidth]{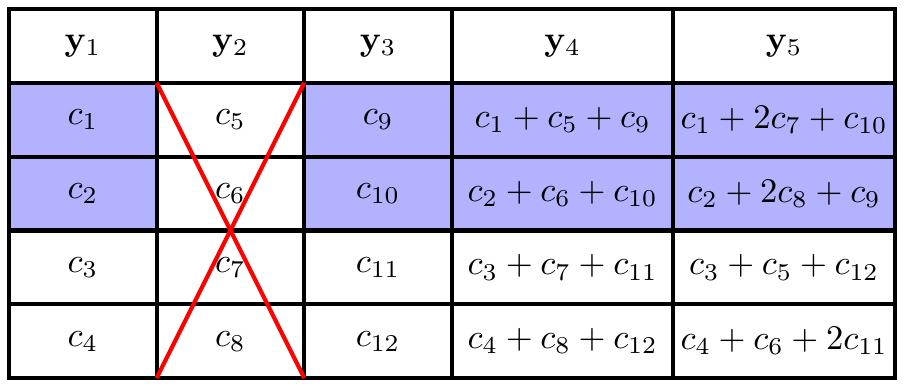}
                \caption{}
                \label{fig:zigzag_des}
        \end{subfigure}%
        \qquad
        \begin{subfigure}[b]{0.47\textwidth}
                \centering
                \includegraphics[width=\columnwidth]{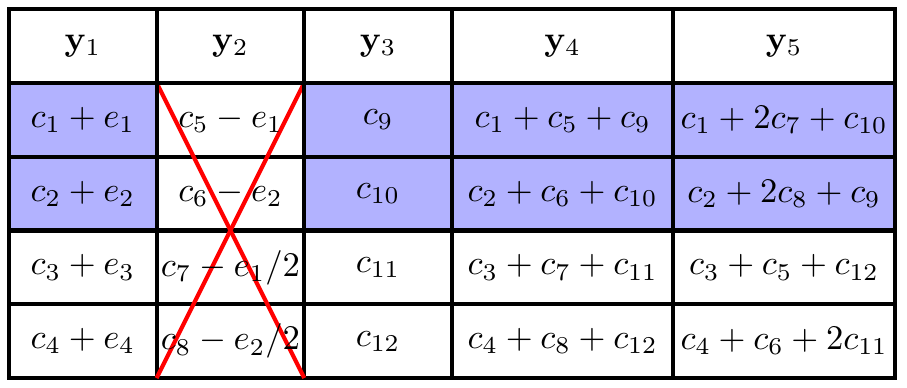}
                \caption{}
                \label{fig:zigzag_repair}
        \end{subfigure}
        \caption{Illustration of the second node repair process in $(5,3)$ Zigzag code: (a) for error free system, (b) for system with erroneous information at the first storage node.}\label{fig:zigzag}
\end{figure*}

Let $\mathbf{x}=(\mathbf{x}_1,\mathbf{x}_2,\ldots,\mathbf{x}_k)\in \F^{\alpha k}$, $k=\frac{\cM}{\alpha}$, be an information vector (a file), $\mathbf{x}_i \in \F^\alpha$ is a block of size $\alpha$, for all $1\leq i\leq k$.
These $k$ blocks are encoded into $n$ encoded blocks $\mathbf{y}_i\in \F^\alpha$, $1\leq i\leq n$, stored in $n$ nodes
 of size $\alpha$,  in the following way:
$$\mathbf{y}=\mathbf{x}\mathbf{G},
$$
where $\mathbf{y}=(\mathbf{y}_1,\mathbf{y}_2,\ldots,\mathbf{y}_n)$ and the generator matrix $\mathbf{G}$  of an MSR code is an $k\times n$ block matrix over $\F$ with blocks of size $\alpha\times\alpha$ given by:
\begin{small}
$$\mathbf{G}=\left[\begin{array}{cccc}
            A_{1,1} & A_{1,2} & \ldots & A_{1,n}  \\
            A_{2,1} & A_{2,2} & \ldots & A_{2,n}  \\
            \vdots & \vdots & \ddots & \vdots \\
            A_{k,1} & A_{k,2} & \ldots & A_{k,n}
          \end{array}
\right].
$$
\end{small}
Note that any blocks submatrix of $\mathbf{G}$ of size $k\times k$ is of the full rank. For a systematic code, we have $\mathbf{y}_i=\mathbf{x}_i$, for $1\leq i\leq k$, and a parity node $j$, $k+1\leq j\leq n$, stores $\mathbf{y}_{j}=\sum_{i=1}^k\mathbf{x}_iA_{i,j}$.

When a node $i$ fails, each node $j$ in a set of $d$ surviving nodes that are contacted to repair this node sends  a vector of length $\beta$ given by $\mathbf{y}_jV_{j,i}$, where $V_{j,i}\in \F^{\alpha\times\beta}$ is a repair matrix used by node $j$ to perform repair of node~$i$.

\subsubsection{Example of Optimal Repair MDS Array Codes}
\label{subsubsec:MDS examples}
In the following, we present an example of  optimal repair MDS array codes for DSS, which we use in Section~\ref{subsec:constructionMSR} to illustrate our coding scheme.

\begin{example}
\label{ex:zig-zag}
(Zigzag code~\cite{zigzag13}).
This class of optimal repair MDS array codes \cite{zigzag13} is  based on generalized permutation matrices.
For an $[n=5,\cM=12,d_{\min}=3,\alpha=4,\beta=2,d=4]_q$ zigzag code presented in Fig.~\ref{fig:zigzag_des}, first three nodes are systematic nodes which store the data $(c_1, c_2,\ldots, c_{12})$, each one $\alpha=4$ information symbols.  The block generator matrix for this code is given by

\begin{small}
$$
\textbf{G}=\left[\begin{array}{ccccc}
               I & \emph{0} & \emph{0} & I & I \\
               \emph{0} & I & \emph{0} & I & A_2 \\
               \emph{0} & \emph{0} & I & I & A_3
             \end{array}
\right],$$
where $I$ and $\emph{0}$ denote the identity matrix and all-zero matrix of size $4\times 4$, respectively; and
$$A_2 = \left[\begin{array}{cccc}
              0 & 0 & 1 & 0 \\
              0 & 0 & 0 & 1 \\
              2 & 0 & 0 & 0 \\
              0 & 2 & 0 & 0
            \end{array}
\right],~A_3=\left[\begin{array}{cccc}
              0 & 1 & 0 & 0 \\
              2 & 0 & 0 & 0 \\
              0 & 0 & 0 & 2 \\
              0 & 0 & 1 & 0
            \end{array}
\right].
$$
\end{small}
Fig.~\ref{fig:zigzag_des} describes node repair process  for this code. When the second node fails, the newcomer node downloads the symbols from the shaded locations at the surviving nodes.
\end{example}

\subsection{Locally repairable codes}
\label{subsec:LRC}

Locally repairable codes (LRCs) are a family of codes for DSS that allow to reduce the number of nodes participating in the node repair process. LRCs are defined as follows.

%
We say that an $[n, \cM, d_{\min}, \alpha]_{\F}$  vector code $C$ is a $(r, \delta, \alpha)$ \emph{locally repairable code} if for each symbol $\cv_i \in \F^{\alpha}$, $1\leq i\leq n$, of a codeword $\cv = (\cv_1,\ldots, \cv_n)\in C$,
there exists a set of indices $\Gamma(i)$ such that
\begin{itemize}
\item $i\in\Gamma(i)$
\item $|\Gamma(i)|\leq r+\delta - 1$
\item  $d_{\min}(C|_{\Gamma(i)})\geq \delta$,
where $C|_{\Gamma(i)}$ denotes the puncturing of $C$ on the set $[n]\backslash\Gamma(i)$ of coordinates.
\end{itemize}

Note that the last two properties imply that each element $j \in \Gamma(i)$ can be written as a function of a set of at most $r$ elements in $\Gamma(i)$ (not containing $j$).

It was proven in~\cite{RKSV12,KPLV12} that the minimum distance of an  $(r, \delta, \alpha)$ LRC of length $n$ and dimension $\cM$ satisfies
\begin{align}
\label{eq:upp_bound}
&d_{\min}(C) \leq n -
\left\lceil\frac{\mathcal{M}}{\alpha}\right\rceil + 1 - \left(\left\lceil\frac{\mathcal{M}}{r\alpha}\right\rceil-1\right)(\delta - 1).
\end{align}

We say that an $(r,\delta, \alpha)$ LRC for an $(n, n-d_{\min}+1)$-DSS is \emph{optimal} if its minimum distance $d_{\min}$ attains the bound~(\ref{eq:upp_bound}).
 In this paper we consider the construction of optimal $(r, \delta, \alpha)$ LRCs  from ~\cite{SRV12} (for $\alpha=1$) and its generalization from~\cite{RKSV12}. This construction is based on concatenation of Gabidulin codes and MDS array codes:
 a file $\fv$ over $\F=\F_{q^N}$ of size $\cM \geq r\alpha$ is first encoded using an $[N,\cM,D]_{q^N}$ Gabidulin code. The codeword  of the Gabidulin code is then partitioned into local disjoint groups, each of size $r\alpha$, and each local group is then encoded using an $[(r+\delta-1),r,\delta,\alpha]_q$ MDS array code over~$\mathbb{F}_q$. (If $r\alpha \nmid N$, there is a group of size $r'\alpha <r\alpha$, which is then encoded by using an $[(r'+\delta-1),r',\delta,\alpha]_q$ MDS array code). A code obtained by this construction is optimal if $ (r+\delta-1)|n$,  $N= \frac{n r \alpha}{r+\delta-1}$ and $q\geq (r+\delta-1)$. For a case when $ (r+\delta-1)\nmid  n$, see the details in~\cite{RKSV12}.

\subsubsection{Example of Optimal LRC}
\label{subsubsec:LRC example}
The following example of  optimal LRC will be used further for illustration of our error-correcting coding scheme.

\begin{example}
\label{ex:array}
We consider a DSS with the following parameters:
$$(\cM,n,r,\delta,\alpha)=(28,15,3,3,4).
$$

Let $\fv$ be a file of size $28$ over $\F_{q^{36}}$, $q\geq 5$.
Let $N=\frac{15\cdot 3\cdot 4}{3+3-1}=36$ and $(a_1,\ldots, a_{12}, b_1,\ldots, b_{12}, c_1,\ldots, c_{12})$ be a codeword  of a $[36,28,9]_{q^{36}}$ Gabidulin code $\cC^{\textmd{Gab}}$, which is obtained by encoding $\cM = 28$ symbols over $\F=\F_{q^{36}}$ of the original file. The Gabidulin codeword is then partitioned into three groups $(a_1,\ldots, a_{12})$, $(b_1,\ldots, b_{12})$, and $(c_1,\ldots, c_{12})$. Encoded symbols in each group are stored on three storage nodes as shown in Fig.~\ref{fig:construction}. In the second stage of encoding, a $[5,3\cdot4,3,4]_{q}$  MDS array code over $\F_{q}$  is applied on each local group to obtain $\delta-1 = 2$ parity nodes per local group. The coding scheme is illustrated in Fig.~\ref{fig:construction}.
\begin{figure}[h]
 \centering
 \includegraphics[width=.8\columnwidth]{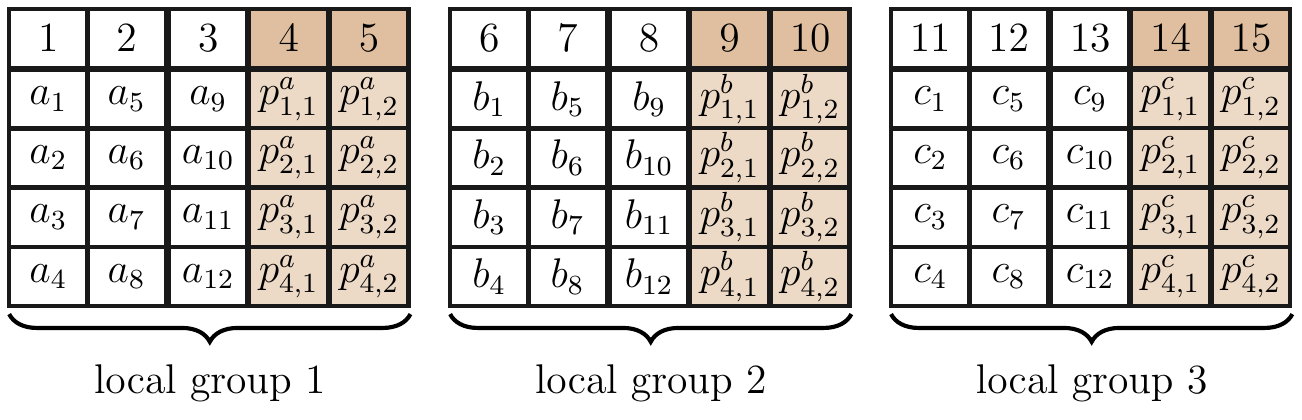}
 \caption{Example of an $(r = 3, \delta = 3, \alpha = 4 )$ LRC with $n = 15$ and $d_{\min}=5$.} \label{fig:construction}
\end{figure}

By~(\ref{eq:upp_bound}) we have $d_{\min}(C^{\rm{loc}})\leq 5$. One can check that every $4$ nodes failures (which is equivalent to at most $8$ rank erasures~\cite{RKSV12}) can be corrected by this code, and thus it has minimum distance $5$.
In addition, when a single node fails, it can be repaired by using the data stored on any three other nodes from the same group.
\end{example}

\subsection{Adversarial errors models}

We assume the presence of an omniscient active adversary
 which can observe all nodes of $(n,k)$-DSS and has full knowledge of the coding scheme employed by the system but can modify at most $t$ nodes, $2t<k$, similarly to~\cite{pawar11}.

Our goal is to design a coding scheme for an $(n,k)$-DSS that allows to deliver the original data to a data collector  even in presence  of $t$ nodes ($t< k/2$) which are fully controlled by an adversary.

In this paper we consider two classes of adversarial attacks:
\begin{enumerate}
\item {\em Static errors:}  an  adversary replaces the content of each affected node with nonsensical information \emph{only once}. The affected node uses this \emph{same} polluted information during all subsequent repair and data collection processes.
\item {\em Dynamic errors:} an  adversary may replace the content of an affected node \emph{each time} the node is asked for the data during data collection or repair process.
\end{enumerate}

As it was mentioned in Introduction,
static errors represent a common type of data corruption due to wear out of storage devices, such as latent disk errors or other physical defects of the storage media, where the data stored on a node is permanently distorted. 
Dynamic errors represent the most general scenario and
hence is more difficult to manage as compared to static errors.

The upper bound on the amount of data that can be stored  reliably on an $(n,k)$-DSS with $t < \frac{k}{2}$ general corrupted nodes (i.e., $t$ dynamic errors) which employs a regenerating code was presented by Pawar et al.~\cite{pawar11}. This amount of data is called \textit{resilience capacity} and denoted by $C_t(\alpha, \beta,d)$.  The bound on the resilience capacity is given by
\begin{equation}\label{bound}
C_t(\alpha, \beta,d)\leq \sum_{i=2t+1}^{k}\min\{(d-i+1)\beta,\alpha\}.
\end{equation}


\begin{remark}
\label{rm:static_vs_general}
Although the static error model is less general than
the  model considered in~\cite{pawar11}, the upper bound in~(\ref{bound})
applies to the static model as well. Pawar et al.~\cite{pawar11}
obtain the upper bound in \eqref{bound} by evaluating a cut of the information flow graph
corresponding to a particular node failure sequence,
pattern of nodes under adversarial attack and data collector.
This information flow graph is also valid in the context of the static error model.
Consequently, its cut which provides an upper bound on the information
flow and represents the amount of data that can be reliably stored on
DSS is applicable to the static error model as well.
As shown in this paper, this bound is tight for static model at the MSR point,
which might not be the case for the general error model considered
in~\cite{pawar11}.
\end{remark}

\section{Coding Schemes for Static Error Model}
\label{sec:static}

In this section we consider the static error model. First, we
present our coding scheme for DSS which employ optimal repair  MDS array codes. Next, we provide an upper bound on the resilience capacity for locally repairable codes and present a coding scheme for DSS which employ optimal locally repairable codes. We prove the error tolerance of the proposed schemes under the \emph{static error model}.  We illustrate the idea by using examples from Section~\ref{subsubsec:MDS examples} and Section~\ref{subsubsec:LRC example} and prove that these constructions are optimal for the static error~model. Finally, we consider error resilience in LRCs with local regeneration.


\subsection{Error resilience in optimal repair MDS array codes}
\label{subsec:constructionMSR}

\subsubsection{Construction of Error-Correcting Optimal Repair MDS Array Codes}
\label{subsubsec:constructionMSR}

First, we present a concatenated coding scheme which is based
on a Gabidulin code as an outer code and an optimal repair MDS array code as an inner code:

\textbf{Construction I.}
Let  $n,k,\alpha,\beta,d$ be the given parameters and let
$\cM, N, m$ be the positive integers such that
$N=\alpha k$ and $\cM\leq N\leq m$.
Consider a file $\fv$ over $\F=\F_{q^m}$ of size $\cM$. We encode the file in two steps before storing it on an $(n,k)$-DSS. First, the file is encoded using an $[N, \cM,D=N-\cM+1]_{q^{m}}$ Gabidulin code $\cC^{\textmd{Gab}}$, with $N$ evaluation points $g_1,\ldots,g_N$ from $\F$, which are linearly independent over $\F_q$, as described in Section~\ref{subsec:rank-metric codes}. Second, the codeword $\textbf{c}_{\fv}=(f(g_1),\ldots,f(g_N))\in \cC^{\textmd{Gab}}$ that corresponds to the file $\fv$  is  encoded using an
$[n, N, d_{\min}, \alpha, \beta,d]_{\F_q}$ optimal repair MDS array (systematic)
 code $C^{\textmd{MDS}}$ over $\F_q$, where $N=\alpha k$, $d_{\min}=n-k+1$, and
$\beta=\frac{\alpha}{d-k+1}$, as described in Section~\ref{subsec: regenerating codes}: the codeword $\textbf{c}_{\fv}$ is partitioned into $k$ blocks of size $\alpha$, which are encoded into $n$ blocks of size $\alpha$ (over $\F$) and then stored on $n$ system nodes.

Note, that we use the optimal repair MDS array code defined $C^{\textmd{MDS}}$
over $\F_q\subseteq \F$, i.e., its generator matrix is over $\F_q$, and during the process of node
repair, a set of surviving nodes transmits linear combinations of the stored elements  with the coefficients from $\F_q$.

The following lemma is useful to show the properties of the constructed code.

\begin{lemma}
\label{lm:equivalent gabidulin code}
Any $k$  different nodes of the proposed scheme contain evaluation of the linearized polynomial $f(x)$ (corresponding to the given file) at $\alpha k$ elements of $\F$ which are linearly independent over $\F_q$. In other words, the content of any $k$ nodes corresponds to a codeword of an $[N=\alpha k,\cM, D=N-\cM+1]_{\F}$ Gabidulin code $\widehat{\cC}^{\textmd{Gab}}$.
\end{lemma}

\begin{proof}
Let $g_1,\ldots, g_N$ be the evaluation points of the Gabidulin code $\cC^{\textmd{Gab}}$ used in the construction.
Consider a set $S=\{i_1,\ldots i_k\}$ of nodes. The content of these nodes is given by $\textbf{c}_{\fv}
G_S$, where $G_S=(a_{i,j})_{i,j=1}^{\alpha k}$ is a $\alpha k\times \alpha k$ submatrix of the generator matrix $\textbf{G}$ of $C^{\textmd{MDS}}$ which corresponds to the nodes in $S$, i.e., which consists of columns $\{(i_j-1)\alpha+1,\ldots, i_j\alpha\}_{j=1}^k$.
The $j$th coordinate of the vector  $\textbf{c}_{\fv}G_S$ of length $\alpha k$ is given by
$(\textbf{c}_{\fv}G_S)_j=(f(g_1),\ldots,f(g_{\alpha k}))(a_{1,j},\ldots,a_{\alpha k,j})^T=\sum_{i=1}^{\alpha k}a_{i,j}f(g_i)=f(\sum_{i=1}^{\alpha k}a_{i,j}g_i)$, where $(a_{1,j},\ldots,a_{\alpha k,j})^T$ is the $j$th column of $G_S$ and the last equality follows from the $\F_q$-linearity of $f$.
Now consider the vector of the new evaluation points of $f$: $(\widehat{g_1},\ldots \widehat{g}_{\alpha k})=(\sum_{i=1}^{\alpha k}a_{i,1}g_i,\ldots,\sum_{i=1}^{\alpha k}a_{i,\alpha k}g_i)=(g_1,\ldots, g_{\alpha k})G_S$.

Note that since $C^{\textmd{MDS}}$ is an MDS array code, then $G_S$ is a full rank matrix over $\F_q$.
Therefore, $\{\widehat{g_1},\ldots \widehat{g}_{\alpha k}\}$ are linearly independent over $\F_q$ if and only if $\{g_1,\ldots, g_{\alpha k}\}$ are linearly independent over $\F_q$. Therefore, the observations $\cv_{\fv}G_{S}$ are essentially evaluations of linearized polynomial at $k\alpha$ linearly independent points over $\F_q$ from $\F$, which correspond to a codeword of an $[N=\alpha k,\cM, D=N-\cM+1]_{\F}$ Gabidulin code $\widehat{\cC}^{\textmd{Gab}}$.
Note that $\widehat{\cC}^{\textmd{Gab}}$ has the same parameters as ${\cC}^{\textmd{Gab}}$, however, the evaluation points of these two codes are different.
\end{proof}

The following theorem shows that if $2t\alpha+1\leq D$, then the proposed scheme
tolerates up to $t$ erroneous nodes, i.e., from any $k$ nodes a data collector can retrieve the original data even in presence of an adversary which controls (modifies) $t$ nodes. Note that the node repairs are performed exactly in the same way as in an error-free DSS.

\begin{theorem}
\label{thm:resilience MSR}
Let $t$ be the number of erroneous nodes in the system based on concatenation of Gabidulin and optimal repair MDS array codes from Construction I.
If $2t\alpha+1\leq D$, then the original data can be recovered from any $k$ nodes.
\end{theorem}

\begin{proof}
Let $\mathbf{c}_{\fv}\in \F$ be the codeword
in $\cC^{\textmd{Gab}}$ which corresponds to the file $\fv$, and let
$(\mathbf{x}_1,\mathbf{x}_2,\ldots,\mathbf{x}_k)$, $\textbf{x}_i\in\F^{\alpha}$, be
the partition of $\textbf{c}_{\fv}$ into $k$ parts of size $\alpha$ each.
Let $(\mathbf{y}_1,\mathbf{y}_2,\ldots,\mathbf{y}_n)$, $\textbf{y}_i\in\F^{\alpha}$ be the encoded by $C^{\textmd{MDS}}$ blocks stored on $n$ nodes.
Let $S=\{i_1,i_2,\ldots,i_t\}$  be the set of indices  of the erroneous nodes.
Hence  the  $i_j$th node, $i_j\in S$, contains  $\sum_{\ell=1}^{k}\mathbf{x}_{\ell}A_{\ell,i_{j}}+\mathbf{e}^{i_j}$,
where $\mathbf{e}^{i_j}=(e^{i_j}_1,e^{i_j}_2,\ldots,e^{i_j}_{\alpha})\in \F^\alpha$ denotes an adversarial error introduced by the $i_j$th node, and $A_{\ell,i_{j}}\in \F_q^{\alpha \times \alpha}$ are the blocks of the generator matrix of $C^{\textmd{MDS}}$.
When the failed nodes are being repaired, the errors from adversarial nodes propagate
to the repaired nodes. In particular, $\ell$th node, $1\leq \ell \leq n$, contains
$\sum_{j=1}^{k}\mathbf{x}_jA_{j,\ell}+\sum_{j=1}^{t}\mathbf{e}^{i_j}B_{\ell}^{i_j}$,
where $B_{\ell}^{i_j}\in \F_q^{\alpha\times \alpha}$ represents the propagation of error $\mathbf{e}^{i_j}$ and
depends on the specific choice of an optimal repair MDS array code.
Suppose a data collector contacts $k$ nodes indexed by $\cD \subset [n]$ and
downloads $\sum_{j=1}^{k}\mathbf{x}_jA_{j,i}+\sum_{j=1}^t \mathbf{e}^{i_j}B_{i}^{i_j}$ from node $i\in \cD$.
\begin{itemize}
  \item Case 1. If these $k$ nodes are all systematic nodes, then
we obtain $(\mathbf{x}_1,\mathbf{x}_2,\ldots,\mathbf{x}_k)+\mathbf{eB}=\textbf{c}_{\fv}+\mathbf{eB}$, where $\mathbf{B}\in \F_q^{\alpha t\times \alpha k}$ is the blocks matrix with the $(j,\ell)$th block of $\mathbf{B}$ is given by $B_{\ell}^{i_j}$, $1\leq\ell\leq k$, $1\leq j\leq t$, and $\mathbf{e}=(\mathbf{e}^{i_1}, \mathbf{e}^{i_2},\ldots, \mathbf{e}^{i_t})\in \F^{\alpha t}$.
  \item Case 2. If not all the $k$ nodes are systematic, we obtain $\widehat{\textbf{c}_{\fv}}+\mathbf{e\widehat{B}}$, where the $(j,\ell)$th block of  the blocks matrix
$\widehat{B}\in\F_q^{\alpha t\times \alpha k}$   is given by $B_{\ell}^{i_j}$, $\ell\in \cD$, $1\leq j\leq t$, $\mathbf{e}$ is defined as previously, and
  $\widehat{c_{\fv}}$ is the codeword of the Gabidulin code $\widehat{\cC}^{\textmd{Gab}}$ with the same parameters as $\cC^{\textmd{Gab}}$, according to Lemma~\ref{lm:equivalent gabidulin code}.
\end{itemize}

In any case, since $\text{rank}(\mathbf{e})\leq t\alpha$ over $\F_q$, and $D\geq 2t\alpha+~1$, a Gabidulin code  can correct this error; consequently, by applying erasure decoding of the MDS array code (for Case 2) used as inner code the original data can be recovered.
\end{proof}

Now we illustrate the idea of the construction with the help of
the example of the optimal repair MDS array code for an $(5,3)$-DSS, presented in Section~\ref{subsubsec:MDS examples}. We consider the case where an adversary pollutes the information stored on a single storage node and demonstrate  that the rank of the error introduced by the adversary does not increase due to node repair dynamics under  the static error model. Hence, a  data collector can recover the correct original information using  decoders for an MRD code $\cC^{\textmd{Gab}}$ and MDS array code $C^{\textmd{MDS}}$.

\begin{example}

 Let $\cC^{Gab}$ be an $[12,4,9]_{q^{12}}$ Gabidulin code and let $C^{\textmd{MDS}}$ be the $[5,12,3,4,2,4]_q$ zigzag code from Example~\ref{ex:zig-zag}.
 First we encode a file $\fv=(f_0,f_1,f_2,f_3)\in \F_{q^{12}}^4$ into a codeword $\textbf{c}_{\fv}=(c_1,\ldots, c_{12})\in\cC^{\textmd{Gab}}$  by 
 $c_i=f(g_i)$,
 $1\leq i\leq 12$, where $\{g_i,\ldots g_{12}\}\subseteq \F_{q^{12}}$ are linearly independent over $\F_q$ and $f(x)=\sum_{j=0}^{3}f_jx^{q^j}$.
Then we encode $\textbf{c}_{\fv}$ again by using $C^{\textmd{MDS}}$. The
first three systematic nodes of $(5,3)$-DSS store a codeword $\textbf{c}_{\fv}$, i.e., the content stored in $i$th systematic node, $1\leq i\leq 3$, is $\mathbf{y}_i = (c_{4(i-1)+1},\ldots, c_{4i}) \in \F^{4}$.
Let us assume that an adversary attacks the first storage node and introduces erroneous information. The erroneous information at the first node can be modeled as $\mathbf{y}_1 + \mathbf{e} = (c_1, c_2, c_3, c_4)+(e_1, e_2, e_3, e_4)$. Now assume that the second node fails. The system is oblivious to the presence of pollution at the first node, and employs an exact regeneration strategy to reconstruct the second node. The reconstructed node downloads the symbols from the shaded locations at the surviving nodes, as described in Fig.~\ref{fig:zigzag_repair}, and solves a linear system of equations to obtain $(c_5, c_6, c_7, c_8) + (-e_1, -e_2, -2^{-1}e_1, -2^{-1}e_2)$, where $2^{-1}$ denotes the inverse element of $2$ in $\mathbb{F}_q$, $q\geq 3$.
\begin{itemize}
  \item Case 1: First assume that a data collector accesses the first three (systematic) nodes to recover the original data. The data collector now has $\widetilde{\mathbf{c}} = \mathbf{c}_{\fv} + \mathbf{e}[I,~B_2^1,~\mathbf{0}]$, where $I$ and $\mathbf{0}$ are $4\times 4$ identity and zeroes matrices, respectively. Note that
\begin{small}
\begin{equation}
B_2^1 = \left[ \begin{array}{cccc}-1&0&-2^{-1}&0\\0&-1&0&-2^{-1}\\0&0&0&0\\0&0&0&0\\ \end{array}\right].
\end{equation}
\end{small}
  \item Case 2: Assume that a data collector accesses the first, second, and the fourth nodes to recover the original data. In this case, the data collector  has $\widetilde{\mathbf{c}} = \widehat{\mathbf{c}_{\fv}} + \mathbf{e}[I,~B_2^1,~\mathbf{0}]$, where $I$, $\mathbf{0}$, and $B_2^1$ are defined as previously, and $\widehat{\mathbf{c}_{\fv}}=(c_1,c_2,\ldots, c_8, c_1+c_5+c_9,c_2+c_6+c_{10},c_3+c_7+c_{11},c_4+c_8+c_{12})$ is a codeword of a $[12,4,9]_{q^{12}}$ Gabidulin code $\widehat{\cC}^{\textmd{Gab}}$ with evaluation points $\{g_1,g_2,\ldots, g_8,g_1+g_5+g_9,g_2+g_6+g_{10},g_3+g_7+g_{11},g_4+g_8+g_{12}\}$.
\end{itemize}

In any case, $\widetilde{\mathbf{c}}$ contains an error of rank at most $4$. Since the minimum rank distance of codes $\cC^{\textmd{Gab}}$ and $\widehat{\cC}^{\textmd{Gab}}$ is $9$, the codeword $\mathbf{c}\in \cC^{\textmd{Gab}}$ (for Case 1) and the codeword $\widehat{\mathbf{c}}\in \widehat{\cC}^{\textmd{Gab}}$ (for Case~2) can be correctly decoded. Now, this allows the original information $\fv$ to be recovered by~(\ref{eq:rank error}) (and by applying the erasure decoding of $C^{\textmd{MDS}}$ in Case 2).
\end{example}
\subsubsection{Optimality of Construction I}
\label{subsubsec: parameters}

 Next, we show that our concatenated scheme attains the upper  bound~(\ref{bound}) on resilience capacity and thus, is optimal (see Remark~\ref{rm:static_vs_general}).

 First, we note that for an $[n,\cM,d_{min},\alpha, \beta,d]_{\F}$ optimal repair MDS code it holds that $\beta=\frac{\alpha}{d-k+1}$ and hence the upper bound~(\ref{bound}) on the resilience capacity can be rewritten as
\begin{align}
\label{bound 2}
C_t(\alpha, \beta,d)&\leq \sum_{i=2t+1}^{k}\min\{(d-i+1)\frac{\alpha}{d-k+1},\alpha\} = \alpha(k-2t).
\end{align}

Let the set of parameters $\{\cM, n,k,\alpha,\beta,d, N, m, D\}$ be as described in the construction of Section~\ref{subsubsec:constructionMSR}.
Then $N=\cM+D-1$ and  $N=\alpha k$.
Let $t$ be an integer such that $D = 2t\alpha+1$. Then $\alpha k=\cM+D-1=\cM+2t\alpha $, and hence $\cM=\alpha (k-2t)$.
Thus, our concatenated coding scheme  achieves the bound in (\ref{bound 2}).
\begin{remark}
The authors  in~\cite{pawar11} provided an explicit construction of codes that attain the bound in \eqref{bound 2} for bandwidth-limited regime. However, this construction has practical limitations for large values of $t$ since the decoding algorithm presented in~\cite{pawar11} is exponential in $t$. On the other hand, the decoding of codewords in the construction presented in our paper  is efficient as it is based on two efficient decoding algorithms: one for an optimal repair MDS array code, and  other one for a Gabidulin code. However, our coding scheme provides resilience for a weaker model of adversarial errors.
\end{remark}

\begin{remark}
In \cite{RSRK_isit12}, Rashmi et al. consider a scenario, referred as `\emph{erasure}',  where some nodes which are supposed to provide data during node repair become unavailable. It is easy to see from~(\ref{eq:rank error})
that our construction can also correct such erasures, as long as the minimum distance of the MRD code used as an outer code is large enough. The codes obtained by our construction also attain the bound on the capacity derived in~\cite{RSRK_isit12}. Here, we note that while our construction works with any MSR code and in particular with an MSR code with high rate, it provides a solution for  a restricted error model, i.e., static~errors.
\end{remark}

\subsection{Error resilience in optimal LRCs}
\label{subsec:static LRC}
In this subsection we study DSS which employ locally repairable codes in the presence of  errors. We consider only optimal LRCs, i.e., we consider an $(n,n-d_{\min}+1)$-DSS where $d_{\min}$ attains the upper bound~(\ref{eq:upp_bound}). First, we provide an upper bound on the resilience capacity and then we present a coding scheme which attains this bound for the static errors.

\subsubsection{Resilience Capacity for LRCs}
\label{subsubsec:resilienceLRC}
In the following, we derive an upper bound on the amount of data that can be reliably stored on a DSS employing an LRC which contains corrupted nodes.

We denote by $C_{t}(r,\delta,\alpha)$ the resilience capacity of an $(n,n-d_{\min}+1)$-DSS which employs an $(r,\delta, \alpha)$ LRC in the presence of $t$ corrupted nodes (for any type of errors).
In other words, a data collector contacting any $n - d_{\min} + 1$ storage nodes can reconstruct the original data of size $C_{t}(r,\delta,\alpha)$ stored on DSS in the presence of at most $t$ adversarial storage nodes.

\begin{theorem}
\label{thm:capacityLRC}
 Consider an $(n,n-d_{\min}+1)$-DSS which employs an $(r,\delta, \alpha)$ LRC. If there are at most $t$ corrupted nodes,
then the upper bound on the resilience capacity is given by
\begin{equation}
\label{eq:resilienceLRC}
C_{t}(r,\delta,\alpha)\leq \left(\rho r-2t\right)\alpha+\min\{h\alpha, r\alpha\},
\end{equation}
where $\rho = \floorb{\frac{n - d_{\min} + 1}{r + \delta - 1}}$,  $h = (n - d_{\min} + 1) - \rho(r + \delta - 1)$,
and $2t<\rho r+\min\{h, r\}$.
\end{theorem}

 \begin{proof}
We assume that the LRC under consideration has $g$ disjoint local groups, as the upper bound on the minimum distance for LRCs is achievable only if the local groups of the code are disjoint~\cite{KPLV12}. We use vector $\tauv = (\tau_1, \tau_2,\ldots, \tau_g)$ to denote an adversarial node pattern, where $\tau_i$ denotes the number of adversarial nodes in $i$th local group and $\sum_{i = 1}^{g}\tau_i = t$.
 Note that during a node repair  a newcomer node contacts any $r$ out of $r + \delta - 2$ surviving nodes in its local group and downloads all the data stored on these $r$ nodes in order to regenerate the failed node.
To obtain the upper bound, we evaluate the value of a cut in the information flow graph for this DSS.  Similarly to~\cite{RKSV12}, we consider a data collector which contacts $r + \delta - 1$ nodes from each of first $\rho = \floorb{\frac{n - d_{\min} + 1}{r + \delta - 1}}$ local groups and $h = (n - d_{\min} + 1) - \rho(r + \delta - 1)$  nodes from $(\rho + 1)$th local group. Then we have the following cut value $\rm CUT$
\begin{align}
{\rm CUT} =  \rho r \alpha + \min\{h\alpha, r\alpha\}.
\end{align}

Next, we consider the adversarial node pattern $\tauv = (\tau_1 = r, \tau_2 = r, \ldots, \tau_{\floorb{t/r}} = r, \tau_{\floorb{t/r} + 1} = t - r\floorb{\frac{t}{r}}, \tau_{\floorb{t/r} + 2} = 0,\ldots, \tau_{g} =  0)$. We further assume that in each local group we have repaired $r + \delta - 1 - r = \delta - 1$ nodes by contacting remaining $r$ nodes and adversarial nodes always belong to non-repaired nodes.
Further, we divide the independent symbols contributing to the value of ${\rm CUT}$ into three groups:
\begin{enumerate}
\item $\Mm_1$: $t\alpha$ symbols corresponding to $t$ adversarial nodes.
\item $\Mm_2$: $t\alpha$ symbols from $t$ intact nodes.
\item $\Mm_3$: ${\rm CUT} - 2t\alpha$ remaining symbols after excluding $\Mm_1$ and $\Mm_2$ from ${\rm CUT}$.
\end{enumerate}

Next, we use an argument similar to that used in~\cite{pawar11}. Note that $\Mm_1$ and $\Mm_2$  can not carry any information to the data collector: Without the knowledge of the identity of adversarial nodes, the data collector can not decide which one of two sets of symbols $\{\Mm_1, \Mm_3\}$ or $\{\Mm_2, \Mm_3\}$ corresponds to legitimate information. Therefore, the data collector relies on symbols associated with $\Mm_3$ to reconstruct the original file. This gives us the following bound on the resilience capacity:
\begin{align*}
C_{t}(r,\delta,\alpha) \leq {\rm CUT} - 2t\alpha = (\rho r - 2t)\alpha + \min\{h\alpha, r\alpha\}.
\end{align*}
\end{proof}

\subsubsection{Construction of Error-Correcting Optimal LRC}
In this subsection we present a concatenated coding scheme, also based on a Gabidulin  and MDS array codes, for error-correcting locally repairable codes and prove the optimality of this construction for static errors model.

\textbf{Construction II.}
Let $n,d_{\min}, r,\delta,\alpha$ be the given parameters. (We assume for simplicity that $(r+\delta-1)|n$).
Let $\cM,\cM',N, m, \rho, h$ be the positive integers such that $\cM'\leq \cM\leq N\leq m$, $r\alpha\leq \cM=\rho r\alpha+\min\{h,r\}\alpha$, $N=\frac{nr\alpha}{r+\delta-1}$,
$\rho = \floorb{\frac{n - d_{\min} + 1}{r + \delta - 1}}$, $h = (n - d_{\min} + 1) - \rho(r + \delta - 1)$,  and $d_{\min}$ attains  the bound~(\ref{eq:upp_bound}), i.e., $d_{\min}= n -
\left\lceil\frac{\mathcal{M}}{\alpha}\right\rceil + 1 - \left(\left\lceil\frac{\mathcal{M}}{r\alpha}\right\rceil-1\right)(\delta - 1)$.

Consider a file $\fv$ of size $\cM'$ over $\F=\F_{q^m}$. The encoding is identical to the encoding for LRC, presented in Subsection~\ref{subsec:LRC}, the only difference is that we apply  an $[N,\cM',D=N-\cM'+1]_{q^m}$ Gabidulin code $\cC^{\textmd{Gab}}$ in the first step of the encoding:
Let $\textbf{c}_{\fv}=(f(g_1),\ldots,f(g_N))$ be the codeword of $\cC^{\textmd{Gab}}$ that corresponds to the file $\fv$, where $g_1,\ldots,g_N\in\F$ are evaluation points for the Gabidulin code, which are linearly independent over $\F_q$. This codeword is partitioned into local disjoint groups, each one of size $r\alpha$ and then each local group is encoded using an $[(r+\delta-1), r, \delta, \alpha]_q$ MDS array code (over~$\F_q$).

The following theorem shows the error resilience of the proposed code. Note, that the node repairs of this scheme are performed exactly in the same (local) way  as in an error-free LRC.
\begin{theorem}
\label{thm:LRC resilience}
Let $t$ be the number of erroneous nodes in the system which employs a $d_{\min}$-optimal LRC based on concatenation of Gabidulin and  MDS array codes, from Construction~II.
If  the minimum distance $D$ of the underlying Gabidulin code satisfies $D\geq 2t\alpha+(\frac{n}{r+\delta-1}-\left\lfloor\frac{n-d_{\min}+1}{r+\delta-1}\right\rfloor)r \alpha-\min\{h\alpha, r\alpha\}+1$, then the original data can be recovered from any $n-d_{\min}+1$ nodes.
\end{theorem}

\begin{proof}
We consider the worst case data collector which contacts $n-d_{\min}+1$ nodes which belong to $\rho$ (or $\rho+1$) different groups, $(r+\delta-1)|(n-d_{\min}+1)$ (or $(r+\delta-1)\nmid  (n-d_{\min}+1)$ ), which contain all the corrupted nodes. During the node repairs, an error spreads to at most all the nodes which belong to the same local group that contains an erroneous node. Based on Lemma~\ref{lm:equivalent gabidulin code} the data collector obtains $\left\lfloor\frac{n-d_{\min}+1}{r+\delta-1}\right\rfloor r \alpha+\min\{h\alpha, r\alpha\}$ symbols of the corresponding Gabidulin code, which can contain an error of rank at most $t\alpha$, similarly to the proof of Theorem~\ref{thm:resilience MSR}. Therefore, this corresponds to $t\alpha$ rank errors and $N-\left\lfloor\frac{n-d_{\min}+1}{r+\delta-1}\right\rfloor r \alpha-\min\{h\alpha, r\alpha\}$ rank erasures. Thus the statement of the theorem follows from the fact that $N=\frac{n}{r+\delta-1} r \alpha$,

\end{proof}
Next, we will show that our error-correcting LRC attains the upper bound on the resilience capacity of Theorem~\ref{thm:capacityLRC}, for a static errors model.

\begin{corollary} Let $\cC^{\emph{LRC}}$ be the LRC for $(n,n-d_{\min}+1)$-DSS obtained from Construction II and let $t$ be the number of corrupted nodes, where $2t<\rho r+\min\{h, r\}$, for $\rho$ and $h$ defined in Theorem~\ref{thm:capacityLRC}. If the corresponding Gabidulin code   has the minimum distance $D=2t\alpha+(\frac{n}{r+\delta-1}-\rho)r \alpha-\min\{h\alpha, r\alpha\}+1$ then $\cC^{\emph{LRC}}$ attains the bound~(\ref{eq:resilienceLRC}) on the resilience capacity.
\end{corollary}
\begin{proof} We need to prove that
$\cM'=\left(\rho r-2t\right)\alpha+\min\{h\alpha, r\alpha\}$. Since the corresponding $[N,\cM',D]_{q^m}$ Gabidulin code is an MRD code, it holds that $\cM'=N-D+1$. Then
\begin{align}
\cM'&=\frac{n r\alpha}{r+\delta-1}-(2t\alpha+\left(\frac{n}{r+\delta-1}-\rho\right)r \alpha -\min\{h\alpha, r\alpha\}+1)+1\nonumber \\
& = \rho r\alpha-2t\alpha+\min\{h\alpha, r\alpha\}. \nonumber
\end{align}

\end{proof}

Now we illustrate the idea of the Construction II.  We consider the case where an adversary pollutes the information stored on a single storage node.
\begin{example}
Consider the code of Example~\ref{ex:array} with additional parameters $t=1$ and $\cM'=20$. The Gabidulin code used in the first step of the encoding is the  $[36,20,17]_{q^{36}}$ code $\cC^{\textmd{Gab}}$ and the MDS array code used in the second step of the construction is the $(5,3)$ zigzag code from Example~\ref{ex:zig-zag}. Here we have $\rho=2$ and $h=1$. In other words, the system stores a file $\fv$ of size $20$ over $\F_{q^{36}}$, and a data collector should reconstruct this file from any $n-d_{\min}+1=15-5+1=11$ nodes.
Let us assume that an adversary attacks the third storage node and introduces erroneous information.
The erroneous information at the third node is modeled as $(a_9, a_{10}, a_{11}, a_{12})+(e_1, e_2, e_3, e_4)$. Now assume that the second node fails. The system is oblivious to the presence of pollution at the third node, and employs the erasure decoding of the $[5,3\cdot4,3,4]_q$ MDS code to reconstruct the second node: Assume that the reconstructed node downloads  all the symbols from the first, the third and the fifth nodes
and solves a linear system of equations to obtain $(a_5, a_{6}, a_{7}, a_{8}) + (-e_4, -2e_3, -e_2, -2^{-1}e_1)$, where $2^{-1}$ denotes the inverse element of $2$ in $\mathbb{F}_q$, $q\geq 3$.

Assume that a data collector contacts $11$ first nodes ($\rho=2$ full groups, 5 nodes in each one, and $h=1$ node in the additional group). He obtains $12+12+4$ symbols corresponding to the Gabidulin codeword $\cv_{\fv}$ (which is of length $36$), where these $28$ symbols contain an error of rank at most~4:
 $$
 \cv_{\fv}+(e_1,e_2,e_3,e_4)[\mathbf{0},B_2,I_{4},\mathbf{0}_{4\times 24}]+(c_5,\ldots,c_{12})[\mathbf{0}_{8\times 28}, -I_8],
 $$
 where $I_m$ denotes the identity matrix of order $m$, $\mathbf{0}_{a\times b}$ denotes $a \times b$ matrix with all of its entries equal to zero, and $B_2$ is defined as follows:
 \begin{small}
\begin{equation*}
B_2 = \left[ \begin{array}{cccc}0&0&0&-2^{-1}\\0&0&-1&0\\0&-2&0&0\\-1&0&0&0\\ \end{array}\right].
\end{equation*}
\end{small}

  In other words, we have $8$ erasures in the Gabidulin codeword and at most 4 rank errors. Since the distance of $\cC^{\textmd{Gab}}$ is $D=17$, by~(\ref{eq:rank error}) the original data is recovered by applying errors and erasures correction of the Gabidulin code. (Note that $(e_1,e_2,e_3,e_4)$, $[\mathbf{0},B_2,I,\mathbf{0}_{4\times 24}]$, $(c_5,\ldots,c_{12})$ are unknown to the decoder but $[\mathbf{0}_{8\times 28}, -I_8]$ is known to the decoder).
\end{example}

\subsection{Error resilience in LRCs with local regeneration}

In this section, we discuss the hybrid codes which for a given locality parameters minimize the repair bandwidth.
These codes, as proposed in~\cite{RKSV12} and \cite{KPLV12}, are obtained by combining locally repairable codes with regenerating codes.

In a  repair process for an original LRC, a newcomer node contacts $r$ nodes in its
local group and downloads \emph{all} the data stored on these nodes.
 To allow a reduction in repair bandwidth, the idea of regenerating codes is used for the hybrid codes.
In particular, a newcomer node  contacts any $d\geq r$ nodes  in its local group and
downloads only  $\beta\leq\alpha$ symbols stored on these nodes in order to repair the
failed node, in other words, a regenerating code is applied in each local group.

That is, when an $[r+\delta-1,r,\delta,\alpha,\beta,d]_{\F_q}$ MSR code is applied in each local group instead of
an MDS array code in the second step of construction of LRC presented in subsection~\ref{subsec:LRC},
then the resulting code, denoted by MSR-LRC,  has the maximal minimum distance
(since an MSR code is also an MDS array code), the local minimum storage per node,
and the minimized repair bandwidth. (The details of this construction and its properties can be found in~\cite{RKSV12}.)
Also, when an MBR code is applied in each local group instead of
an MDS array code in the second step of construction of LRC presented in subsection~\ref{subsec:LRC},
then the resulting code, denoted by MBR-LRC,  has the maximal possible minimum distance,
and the local minimum repair bandwidth. (The details of this construction and its properties can be found in~\cite{MBR-LRC}.)

In the following, we consider error resilience for MSR-LRC and for MBR-LRC. In particular, we provide a upper bound on the resilience capacity for these codes and consider construction for error-correcting MSR-LRC and MBR-LRC. We denote by $C_t(r,\delta,\alpha,\beta,d)_{\textmd{MSR}}$ ($C_t(r,\delta,\alpha,\beta,d)_{\textmd{MBR}}$) the resilience capacity of an $(n,n-d_{\min}+1)$-DSS which employs an $(r,\delta,\alpha,\beta,d)$ MSR-LRC (MBR-LRC) in the presence of $t$ corrupted nodes.

\begin{theorem}
\label{thm:capacity MSR-LRC and MBR_LRC}
 Consider an $(n,n-d_{\min}+1)$-DSS which employs an $(r,\delta, \alpha,\beta,d)$ MSR-LRC (MBR-LRC), for $r+\delta-1>d>r$. If there are at most $t$ corrupted nodes, then the upper bound on the resilience capacity of MSR-LRC is given by
\begin{align}
\label{eq:resilience MSR-LRC}
C_{t}(r,\delta,\alpha,\beta,d)_{\textmd{MSR}} &\leq \left(\rho - 2\floorb{\frac{t}{d}}\right)r\alpha +\left(\min\{h, r\} - 2\min\{\gamma,r\} \right)\alpha, \nonumber
\end{align}
and the resilience capacity for the MBR-LRC is upper bounded by
\begin{align*}
C_{t}(r,\delta,\alpha,\beta,d)_{\textmd{MBR}}&\leq \min\{\textrm{term I}, \textrm{term II}\},
\end{align*}
where
\begin{align*}
\textrm{term I} &  = \left(\rho- 2\floorb{\frac{t}{d}}\right) B_{\rm MBR} - 2 \sum_{i = 1}^{\min\{\gamma, r\}} (d - i + 1)\beta + \sum_{i = 1}^{\min\{h, r\}}(d - i + 1)\beta, \\
\textrm{term II} & =  \tilde{\rho}\sum_{i = 2s+1}^{d}(d - i + 1)\beta + (\rho - \tilde{\rho})\sum_{i = 2\tilde{s}+1}^{d}(d - i + 1)\beta + \sum_{i = 2\hat{s}+1}^{\min\{h,d\}}(d - i + 1).
\beta,
\end{align*}
Here, $\rho = \floorb{\frac{n - d_{\min} + 1}{r + \delta - 1}}$, $h = (n - d_{\min} + 1) - \rho(r + \delta - 1)$,
$\gamma = \left(t - \floorb{\frac{t}{d}}d\right)$, $B_{\rm MBR}=r\alpha-\frac{r(r-1)}{2}\beta$ and $\tilde{\rho}, s, \tilde{s}$, and $\hat{s}$ are such that $\tilde{\rho}s + (\rho - \tilde{\rho})\tilde{s} + \hat{s} = t$, $2\max\{s, \tilde{s}\} \leq d$, and $2\hat{s} \leq h$.
\end{theorem}

\begin{proof}
\emph{MSR-LRC}:
 Similarly to the proof of Theorem~\ref{thm:capacityLRC}, we evaluate the value of a cut in the corresponding information flow graph.
We again consider a data collector that contacts $r + \delta - 1$ nodes from each of first $\rho$ local groups and $h$  nodes from $(\rho + 1)$th local group. We further assume that the pattern of adversarial nodes is $\tauv = (\tau_1 = d, \tau_2 = d, \ldots, \tau_{\floorb{t/d}} = d, \tau_{\floorb{t/d} + 1} = t - d\floorb{\frac{t}{d}}, \tau_{\floorb{t/d} + 2} = 0,\ldots, \tau_{g} =  0)$. In each group, we assume that all but $d$ nodes have been repaired at least once. The remaining $d$ nodes are used to repair all node failures. The corrupted nodes are assumed to be among these $d$ nodes in their local groups. The cut value for the information flow graph associated with this scenario is given by
\begin{align}
{\rm CUT_{\textmd{MSR}}} =  \rho r \alpha + \min\{h\alpha, r\alpha\}.
\end{align}
Similar to the proof of Theorem~\ref{thm:capacityLRC}, we divide the independent symbols in $\rm CUT_{\textmd{MSR}}$ into three groups:
\begin{enumerate}
\item $\Mm_1$: $\floorb{\frac{t}{d}}r\alpha + \min\{\gamma,r\}\alpha$ symbols corresponding to $t$ adversarial nodes.
\item $\Mm_2$: $\floorb{\frac{t}{d}}r\alpha +\min\{\gamma,r\}\alpha$ symbols from $\floorb{\frac{t}{d}}r + \gamma$ intact nodes.
\item $\Mm_3$: $\rm CUT_{\textmd{MSR}} - 2\left(\floorb{\frac{t}{d}}r +\min\{\gamma,r\} \right)\alpha$ remaining symbols after excluding $\Mm_1$ and $\Mm_2$
from $\rm CUT_{\textmd{MSR}}$.
\end{enumerate}
Again following the argument similar to that for LRCs, we get the following bound on the resilience capacity for an MSR-LRC:
\begin{align*}
C_{t}(r,\delta,\alpha,\beta,d)_{\textmd{MSR}} &\leq \rm CUT_{\textmd{MSR}} - 2\left(\floorb{\frac{t}{d}}r + \min\{\gamma,r\} \right)\alpha \nonumber \\
&=\left(\rho - 2\floorb{\frac{t}{d}}\right)r\alpha +\left(\min\{h, r\} - 2\min\{\gamma,r\} \right)\alpha.
\end{align*}
Note that if $t<r$ then this bound is the same as the bound~(\ref{eq:resilienceLRC}) for LRCs.


\emph{MBR-LRC}:
 If we consider the same adversarial node pattern as in case MSR-LRC, we get the following bound on the resilience capacity for an MBR-LRC:
\begin{align}
\label{eq:file_mbr1}
&C_{t}(r,\delta,\alpha,\beta,d)_{\textmd{MBR}} \leq \left(\rho - 2\floorb{\frac{t}{d}}\right) B_{\rm MBR} - 2 \sum_{i = 1}^{\min\{\gamma, r\}} (d - i + 1)\beta + \sum_{i = 1}^{\min\{h, r\}}(d - i + 1)\beta.&
\end{align}

Alternatively, we consider another pattern of eavesdropped nodes $\tauv = (\tau_1 = s, \tau_2 = s,\ldots, \tau_{\tilde{\rho}}= s, \tau_{\tilde{\rho} + 1} = \tilde{s}, \tau_{\rho} = \tilde{s}, \tau_{\rho + 1} = \hat{s}, \tau_{\rho + 2} = 0,\ldots, \tau_{g} = 0)$. Here, $\tilde{\rho}, s, \tilde{s}$, and $\hat{s}$ are such that $\tilde{\rho}s + (\rho - \tilde{\rho})\tilde{s} + \hat{s} = t$, $2\max\{s, \tilde{s}\} \leq d$, and $2\hat{s} \leq \min\{h,r\}$. Note that such choice is always possible given the particular choice of data collector and the assumption that $2t < \rho r+\min\{h,r\}$. For this pattern of adversarial node, we obtain the following upper bound on the resilience capacity for an MBR-LRC:
\begin{align}
\label{eq:file_mbr2}
&C_{t}(r,\delta,\alpha,\beta,d)_{\textmd{MBR}} \leq  (\rho - \tilde{\rho})\sum_{i = 2\tilde{s}+1}^{d}(d - i + 1)\beta +\tilde{\rho}\sum_{i = 2s+1}^{d}(d - i + 1)\beta + \sum_{i = 2\hat{s}+1}^{\min\{h,d\}}(d - i + 1)\beta.
\end{align}
Now, we can pick minimum of RHS of \eqref{eq:file_mbr1} and \eqref{eq:file_mbr2} as an upper bound on the file size.
\end{proof}

\subsubsection{Construction of Error-Correcting LRC with Local Regeneration}

Similarly to the construction of error-correcting LRCs, where the only difference to the construction of $d_{\min}$-optimal LRC (without error correction) is the larger minimum rank distance of the corresponding Gabidulin code, the construction for error-correcting LRCs with local regeneration is based on the construction of an error-free MSR-LRC (MBR-LRC), where the Gabidulin code is chosen with larger minimum distance.
We note that since for the case $t<r$ the bounds for LRCs and MSR-LRCs are identical, the optimal codes can be obtained simply by replacement  of MDS array codes by MSR codes in Construction II. However,   constructions for optimal error-correcting MSR-LRCs for general cases and also for optimal error-correcting MBR-LRC codes  still remain an open problem.


\section{Dynamic Error Model}
\label{sec:dynamic}
In this section, we consider the problem of designing coding schemes for DSS that work under dynamic error model.
Note, that in a static error model each time an attacked node is requested for the data to be sent, it sends some linear combinations of the data that has been modified on it by an adversary, which the adversary is allowed to do only once. Therefore, the rank of the error that a single node under static attack causes throughout the operation of DSS is bounded above by $\alpha$. This is not the case under the dynamic error model as a single attacked node can inject an error of large rank if it is utilized in multiple node repairs, which may render the data stored on DSS useless.


Towards this model, some results are presented in \cite{pawar11} and \cite{RSRK_isit12}. The coding scheme proposed in \cite{pawar11} does not have an efficient decoding during the data reconstruction process and it works specifically with bandwidth efficiently repairable codes at the MBR point.
The coding scheme of \cite{RSRK_isit12} deals with the dynamic error model at the MSR point, but there scheme works only for low rate, i.e., $2k \leq n+1$. 

Next, we present two solutions to deal with attack under the dynamic error model. The first solution aims to correct errors during the node repair process. The second approach is based on existing literature on subspace signatures. All the results presented in this section are given for optimal repair MDS array codes. Since our locally repairable codes make use of MDS array codes (or MSR codes for codes with local regeneration) in each local group, the similar ideas can be applied for this family of DSS codes as well.

\subsection{Na\"{\i}ve scheme for dynamic error model}
A solution for the dynamic error model is to adopt a repair scheme where a newcomer node utilizes the redundancy in the downloaded data to perform error-free exact repair even in the presence of errors in the downloaded data. Next, we analyze the maximum amount of information that can be stored on the DSS employing concatenated codes proposed in Section~\ref{subsubsec:constructionMSR} under the dynamic error model, if an error free node repair is performed.

When a storage node fails, a newcomer node downloads $d\beta$ symbols from any $d$ surviving nodes $(d \geq k)$ (for MSR code $d\beta=\alpha+(k-1)\beta$). Since there can be at most $t$ adversarial nodes present in the system, the newcomer node receives at most $t\beta$ erroneous symbols. Therefore, out of $k\alpha$ symbols of a Gabidulin codeword, by Lemma~\ref{lm:equivalent gabidulin code},  the newcomer has $(k-1)\beta + \alpha$ symbols (using the fact that the inner code is an MDS code and we perform bandwidth efficient repair). All the other $k\alpha-(k-1)\beta -\alpha=(\alpha-\beta)(k-1)$ symbols of a Gabidulin codeword can be considered as the erased symbols. Let $\cM'$ denote the number of information symbols (over $\mathbb{F}_{q^m}$) that are stored on the DSS. Then the minimum distance $D$ of the corresponding Gabidulin code satisfies $D=k\alpha-\cM'+1$. Therefore
we can reconstruct the entire Gabidulin codeword and thus the data stored on the failed node, if we have
$
D=k\alpha - \cM' +1 \geq 2t\beta + (k-1)(\alpha - \beta) +1.
$

This gives us
\begin{equation}
\label{eq:bound2}
\cM' \leq \alpha  + (k-2t-1)\beta.
\end{equation}
Note that the bound in (\ref{bound}) is still applicable. For $k = 2t + 1$, the right hand side expression in (\ref{eq:bound2}) is equal to that in (\ref{bound}). Therefore, this na\"{\i}ve repair scheme is optimal in terms of the capacity of DSS even in the dynamic error model. However, the difference between these bounds is monotonically increasing with $(k-2t-1)$ and the solution proposed in this section is suboptimal for general values of system parameters $k$ and~$t$.

\begin{remark}
In a similar way, it can shown that the LRCs based construction presented in Section~\ref{subsec:static LRC} is optimal under dynamic error model when $\rho r+\min\{h,r\}=2t+1$.
\end{remark}

\subsection{Subspace signatures approach}
\label{subsec:dynamic_errors}

As mentioned previously, in the dynamic error model an attacked node can inject a high rank error. Thus, it is desirable to restrict the rank of the aggregate error that a particular attacked node can cause in the entire system under dynamic error model. In this subsection, we propose to combine the existing literature on detecting subspace pollution with MRD codes to counter a dynamic attack. Next, we illustrate this with the help of subspace signatures proposed in \cite{zhao}.

Let us consider an $n$-nodes DSS that employs a Gabidulin and a bandwidth efficiently repairable code based storage scheme as explained in Section~\ref{subsubsec:constructionMSR}. For node $i$, content stored on it, i.e. $\mathbf{y}_i\in \F^{\alpha}=\F_{q^m}^{\alpha}$, can be viewed as an $m \times \alpha$ matrix over $\mathbb{F}_q$. These $\alpha$ column vectors of length $m$ stored on $i$th node span a subspace (column space of $\mathbf{y}_i$ when viewed as a matrix over $\mathbb{F}_q$) in $\mathbb{F}^m_{q}$ of dimension at most $\alpha$. Since all elements of the coding matrix and repair matrices are from $\mathbb{F}_q$,
during node repair process  node $i$ sends
$\beta$ vectors that lie in the subspace spanned by $\mathbf{y}_i$. If we make sure that even under the dynamic error model an attacked node sends vectors from the same $\alpha$-dimensional subspace of $\mathbb{F}_q^m$ during node repairs and data reconstruction, a data collector encounters at most $t\alpha$-rank error, which can be corrected with a Gabidulin code of large enough distance as in the static model. Subspace signatures solve this problem of enforcing the requirement that a node sends data (vectors) from the same $\alpha$-dimensional subspace of~$\mathbb{F}_q^m$.

 We assume existence of a trusted verifier, who stores all $n$ subspace signatures, one signature for each storage node, generated according to the procedure explained in \cite{zhao}. Whenever a particular node sends data during a node repair or data reconstruction, the truster verifier checks the data against the stored subspace signature corresponding to that particular storage node.

For the purpose of the \emph{data reconstruction}, whenever a node does not pass the signature test, this node is considered as $\alpha$ rank erasures. If $s \leq t$ nodes fail the test during data reconstruction, the data collector deals with $s\alpha$ rank erasures and $(t-s)\alpha$ rank errors.
  Given that the outer Gabidulin code has minimum rank distance $2t\alpha+1 \geq 2(t-s)\alpha + s\alpha +1$, the original data can be reconstructed without an error. 

Next, we argue how subspace signatures help restrict rank of the error introduced in the system during a \emph{node repair process}.
  Assume that node $i$ fails. Let $\mathcal{R}_i \subseteq \{1,\ldots,n\}\backslash\{i\}$ denote the set of $d$ surviving nodes that are contacted to repair node $i$. In order to repair node $i$, each node $j \in \mathcal{R}_i$ is supposed to send $\mathbf{y}_jV_{ji}$, where $V_{ji}$ is an $\alpha \times \beta$ repair matrix of node $i$ associated with node $j$. Since the data downloaded through all the surviving nodes is verified against subspace signatures, data from node $j$ passes the test if it is of the form $\mathbf{y}_j\widehat{V}_{ji}$, where $\mathbf{y}_j\widehat{V}_{ji}$ is in the column space of $\mathbf{y}_j$ and  $\widehat{V}_{ji}$ may be different form $V_{ji}$.

If any of the surviving (helper) nodes does not pass the test, the trusted verifier begins the na\"{\i}ve repair for the failed node and the nodes that fail the test. During this na\"{\i}ve repair, entire data is downloaded from a set of $k-s$ nodes out of $d-s$ nodes that provide data for node repair and pass the subspace test. Here, $s$ is the number of nodes that fail the subspace test. Note that each node of these $k-s$ selected nodes provides additional $\alpha - \beta$ symbols as it has already sent $\beta$ symbols (over $\mathbb{F}_q^{m}$). The decoding algorithm for Gabidulin codes is run on $(k-s)\alpha$ symbols downloaded from this selected set of $k-s$ nodes. There can be at most $t-s$ adversary nodes present in the selected set of $k-s$ nodes ($s$ adversarial nodes that failed the subspace test are excluded from this process), which can contribute at most $(t-s)\alpha$ erroneous symbols. Since the distance of the Gabidulin code is greater than $2(t-s)\alpha + s\alpha +1$, the decoding algorithm recovers the original file, which is used to get the data
stored on nodes being repaired.

In case when all the adversarial nodes pass the test, the data provided by each node $j \in \mathcal{R}_i$ is of the form
\begin{equation*}
\mathbf{y}_j\widehat{V}_{ji} = \mathbf{y}_jV_{ji} + \mathbf{y}_j(\widehat{V}_{ji}-V_{ji}).
\end{equation*}
After performing exact repair process for node $i$, node $i$ stores
$\mathbf{y}_i + \mathbf{y_e}B_i$, where $B_i$ is an $t\alpha \times \alpha$ matrix over $\mathbf{F}_q$ and $\mathbf{y_e} = [\mathbf{y}_{i_1},\ldots, \mathbf{y}_{i_t}]\in\F_{q^m}^{t \alpha}$. Here $\{i_1,\ldots, i_t\}$ denotes the set of $t$ adversarial nodes. After the node repair, the trusted verifier generates a new subspace signature corresponding to the data stored on a node $i$ for future verification. At any point of time, the data stored on DSS can be represented as
\begin{equation}
\label{eq:anytime_dyn}
\widetilde{\mathbf{y}} = \mathbf{y} + \mathbf{y_eB},
\end{equation}
where columns with indices from $\{(i-1)\alpha+1,\ldots, t\alpha\}$ of $\mathbf{B}$ are equal to $\alpha$ columns of $B_i$, $1\leq i\leq n$. It is evident from (\ref{eq:anytime_dyn}) that the rank of the aggregate error in the system is at most $t\alpha$ and a Gabidulin code with large enough distance can ensure the reliable recovery of the original data.


\section{Conclusion and Future Research}
\label{sec:conclusion}

A novel concatenated coding scheme for DSS is presented.
The scheme makes use of rank-metric codes, in particular, Gabidulin codes, as the first step of the process of encoding the data. In the second step of the encoding process, MDS optimal repair array codes (locally or globally) are used. This construction ensures resilience against static adversarial errors.  A modification of the scheme based on subspace signatures enables resilience against dynamic errors. Also, upper bounds on the resilience capacity for LRCs, MSR-LRCs, and MBR-LRCs are presented.

We conclude with a list of open problems for future research.
\begin{enumerate}
  \item Do there exist (explicit) high-rate error-correcting MSR codes which attain the upper bound~(\ref{bound}) on the resilience capacity for a general (dynamic) error model?
  \item Is it possible to improve the bounds on resilience capacity for MSR-LRCs and MBR-LRCs?
  \item Do there exist (explicit) optimal MSR-LRCs and MBR-LRCs for a general set of parameters?
\end{enumerate}



\bibliographystyle{unsrt}
\bibliography{ErrorsDSS}


\end{document}